\documentclass[runningheads,orivec]{llncs}
\usepackage{amsmath,amssymb}
\usepackage{graphicx}
\usepackage{booktabs}
\usepackage{tikz}
\usepackage{tikz-cd}
\usepackage{multicol}
\usepackage{enumitem}
\usepackage{xcolor}
\usepackage{comment}
\usepackage[para]{footmisc}
\usepackage{hyperref}
\makeatletter
\renewcommand{\@makefntext}[1]{%
  \makebox[.5\columnwidth][l]{%
    \makebox[1.5em][r]{\@makefnmark}\hspace{.5em}#1}%
  \hskip0pt\relax}
\makeatother
\usetikzlibrary{fit,positioning}
\tikzset{
  vertexdot/.style={circle,fill=black,draw=black,minimum size=5pt,
    inner sep=0pt},
  vertexlabel/.style={font=\small,inner sep=1pt},
  pidvalue/.style={font=\small,color=red!80!black,inner sep=1pt},
  weightlabel/.style={font=\small,above=1mm}
}
\spnewtheorem{observation}{Observation}{\bfseries}{\itshape}

\makeatletter
\newenvironment{firstdeferredlemma}{%
  \begingroup\def\@thmcounterend{ \textnormal{(*)}\textsuperscript{1}.}%
  \begin{lemma}}{%
  \end{lemma}\endgroup}
\newenvironment{deferredtheorem}{%
  \begingroup\def\@thmcounterend{ \textnormal{(*)}.}\begin{theorem}}{%
  \end{theorem}\endgroup}
\newenvironment{deferredlemma}{%
  \begingroup\def\@thmcounterend{ \textnormal{(*)}.}\begin{lemma}}{%
  \end{lemma}\endgroup}
\newenvironment{deferredproposition}{%
  \begingroup\def\@thmcounterend{ \textnormal{(*)}.}\begin{proposition}}{%
  \end{proposition}\endgroup}

\makeatother
\newcommand{\deferredproofnote}{%
  \begingroup
  \begin{NoHyper}%
  \footnotetext[1]{Due to space constraints, results marked with $\star$ have their proofs deferred to
  the Appendix.}%
  \end{NoHyper}%
  \endgroup}

\title{Perfect Italian Domination on and Near Split Graphs:
Algorithms, Hardness, and Approximation}
\titlerunning{Perfect Italian Domination on and Near Split Graphs}

\author{Anand Babu N B\thanks{\mbox{ORCID: 0000-0002-9904-0488}} \and
Ashwin Jacob\thanks{\mbox{ORCID: 0000-0003-4864-043X}} \and
M. S. Manjusha\thanks{\mbox{ORCID: 0009-0001-2509-3667}} \and
P. Renjith\thanks{\mbox{ORCID: 0000-0002-4934-4571}}}
\authorrunning{A. Babu et al.}

\institute{National Institute of Technology Calicut, Kozhikode, Kerala, India\\
\email{anandbabu@nitc.ac.in, ashwinjacob@nitc.ac.in,\\
manjusha\_p220226cs@nitc.ac.in, renjith@nitc.ac.in}}

\begin{document}
\maketitle

\begin{abstract}
A perfect Italian dominating function assigns a label from $\{0,1,2\}$ to
each vertex so that the labels in the neighborhood of every zero-labeled
vertex sum to exactly two. Although the associated decision problem (PID) is NP-complete on chordal graphs, the complexity status on split graphs remains open. We give an $O(n^4)$ time algorithm for split graphs on $n$ vertices. We note that this is in contrast to Roman, Italian, and perfect Roman domination, which are NP-complete on split graphs. We further show that this tractability boundary is tight: PID is NP-hard on graphs at deletion distance one from split graphs. We also provide dichotomy results for PID on
restricted extensions of split graphs formed by attaching pendant stars, and establish inapproximability results for the same.
\keywords{Split graphs \and Perfect Italian domination \and Minimum-capacity cut
\and Split-deletion distance
\and Approximation hardness }
\end{abstract}

\section{Introduction}\label{sec:intro}
Roman domination originated from a strategy for defending the Roman Empire \cite{roman} and was formalized in graph-theoretic terms by Cockayne et al.~\cite{roman_domination_graphs}. Given a graph $G=(V, E)$, a Roman dominating function $f:V \rightarrow
\{0, 1, 2\}$ is a labeling of vertices such that each vertex $u \in V$ with $f(u) = 0$ has a vertex $v \in N(u)$ with $f(v) = 2$. 

Several variants of Roman domination arise by modifying the neighborhood requirement. In a \emph{perfect
Roman dominating function}, every vertex $u$ with $f(u)=0$ has exactly one
neighbor $v$ with $f(v)=2$ \cite{HENNING2018235}. In an \emph{Italian dominating
function}, introduced as Roman $\{2\}$-domination by Chellali et
al.~\cite{italian_domination}, every vertex $u$ with $f(u)=0$ instead satisfies
$\sum_{v\in N(u)}f(v)\geq2$. Italian domination
has been studied on numerous graph classes; a recent overview is given in
\cite{almulhim2024survey}.
A function
$f:V(G)\rightarrow\{0,1,2\}$ is a \emph{perfect Italian dominating function}
(PID-function) if
$\sum_{v\in N_G(u)}f(v)=2$ for every vertex $u\in V(G)$ with $f(u)=0$.
The weight of $f$ is $\omega(f)=\sum_{v\in V(G)}f(v)$, and the minimum weight
of a PID-function on $G$ is the perfect Italian domination number
$\gamma_I^P(G)$. The corresponding decision problem, denoted by PID, asks
whether $G$ admits a PID-function of weight at most a given integer $k$.

Since its introduction by Haynes and Henning
\cite{perfect_italian_domination}, perfect Italian domination has been studied
on several structured graph families \cite{varghese2021italian,lauri2020perfect},
including cographs \cite{PID_Cograph} and their generalizations
\cite{PID_DCCP45}. Its algorithmic complexity varies sharply across graph
classes: PID is NP-complete on chordal graphs, whereas polynomial-time algorithms
are known for block graphs, series--parallel graphs, distance-hereditary
graphs, threshold graphs, and cographs
\cite{PID_CSPB,PID_DCCP45,lauri2020perfect,PID_Cograph}. This contrast
motivates the search for the precise structural boundary between tractability
and NP-completeness.

Split graphs are a natural next graph class for refining this complexity boundary.
They form a structurally restricted subclass of chordal graphs such that their vertices
can be partitioned into a clique and an independent set, while the adjacencies
between the two parts remain arbitrary. Thus, they retain enough flexibility
to support nontrivial instances while offering sufficient structure for
polynomial-time solvability to remain plausible. Determining the
complexity of PID on split graphs helps locate the boundary more precisely.
Lauri and Mitillos~\cite{lauri2020perfect} showed that the general upper bound
$\gamma_I^P(G)\leq |V(G)|$ is tight for connected split graphs by exhibiting
an infinite family attaining equality. Their extremal construction, however,
does not yield an algorithm for computing an optimum PID-function on an
arbitrary split graph. Moreover, Roman domination,
Italian domination, and perfect Roman domination are all NP-complete on split
graphs \cite{roman_domination_graphs,ID_SB,PRD_S}. These results raise the question of whether PID is equally hard on split graphs. We resolve this question by providing a polynomial time algorithm.

Having established polynomial-time solvability on split graphs, we next ask
whether this tractability persists under controlled structural extensions.
We consider $K_{1,1}$-split graphs, which allow each independent-side center
to have at most one private pendant leaf. More generally, a
$K_{1,r}$-split graph consists of a clique $K$ together with a collection of
vertex-disjoint stars outside $K$. The star centers form a set $I$ and may
have neighbors in $K$, whereas each leaf is adjacent only to its own center. The case $r=0$
recovers split graphs. We determine the complexity of PID on $K_{1,1}$-split graphs and complete the
classification across the strict and non-strict $K_{1,r}$-split graph classes.

The contributions of this paper are as follows.
\begin{enumerate}
    \item We prove that PID is polynomial-time solvable on split graphs. After
    characterizing the possible labelings of the clique, we reduce the
    remaining case to a weighted partition problem and then to a
    minimum-capacity cut problem. The resulting algorithm computes a minimum
    PID-function in $O(n^4)$ time.

    \item We prove that PID is NP-complete even on graphs at deletion distance
    one from split graphs.
    Thus PID parameterized by split-deletion distance is para-NP-hard unless $\mathrm{P}=\mathrm{NP}$.

    \item We prove that PID is NP-complete on strict $K_{1,1}$-split graphs by a reduction from Restricted Exact Cover by 3-Sets. Thus, adding one
    pendant neighbor to every independent-side vertex changes the complexity
    from polynomial time to NP-complete. Since every non-strict
    $K_{1,r}$-split class with $r\geq1$ contains the strict $K_{1,1}$-split
    graphs, PID is NP-complete on every such class for fixed $r\geq1$.

    \item We give a linear-time algorithm for strict $K_{1,r}$-split
    graphs for every fixed $r\geq3$. Together with our hardness result for
    $r=1$ and the known hardness for $r=2$ \cite{PID_CSPB}, this yields a
    complete complexity classification of the strict $K_{1,r}$-split graph
    classes for all fixed $r\geq1$.

    \item We establish approximation bounds for both extensions of split
    graphs. Even at deletion distance exactly one, we rule out a polynomial-time
    $(2-\varepsilon)$-approximation for every fixed $0<\varepsilon<1$,
    unless $\mathrm{P}=\mathrm{NP}$.
    On strict $K_{1,1}$-split graphs, no PTAS exists unless $\mathrm{P}=\mathrm{NP}$, even when $|K|\leq|I|$. On the positive side, we obtain
    constant-factor approximation algorithms for strict $K_{1,r}$-split
    graphs with $r\in\{1,2\}$ and $|K|\leq c|I|$, where $c$ is fixed.
\end{enumerate}

The remainder of the paper is organized as follows. Section~2 introduces the
notation and graph classes used throughout the paper. Section~3 presents the
polynomial-time algorithm for split graphs. Section~4 establishes hardness at split-deletion
distance one. Section~5 proves NP-completeness on strict $K_{1,1}$-split
graphs and presents the linear-time algorithm for strict $r\geq3$. Section~6 presents
the approximation and inapproximability results for both extensions of split
graphs.

\section{Graph Preliminaries}
All graphs considered in this paper are finite, simple, and connected unless
stated otherwise. Let $G=(V(G),E(G))$ be a graph. For $v\in V(G)$, its open
and closed neighborhoods are
$N_G(v)=\{u\in V(G):uv\in E(G)\}$ and $N_G[v]=N_G(v)\cup\{v\}$.
For $S\subseteq V(G)$, let
$N_G(S)=\bigcup_{v\in S}N_G(v)$ and $N_G[S]=N_G(S)\cup S$. The subgraph
induced by $S$ is denoted by $G[S]$. We omit the subscript $G$ whenever the
underlying graph is clear.

A graph is \emph{chordal} if it has no induced cycle of length at least four.
A graph $G$ is a \emph{split graph} if its vertex set admits a partition
$V(G)=K\mathbin{\dot\cup}I$ such that $K$ induces a clique and $I$ induces an
independent set. We call $(K,I)$ a split partition and take $K$ to be maximal
unless stated otherwise. For $v\in K$, define
$N_G^I(v)=N_G(v)\cap I$ and $d_G^I(v)=|N_G^I(v)|$.
For $S\subseteq K$, let $N_G^I(S)=N_G(S)\cap I$ and
$d_G^I(S)=|N_G^I(S)|$.

For an integer $r\geq0$, a \emph{$K_{1,r}$-split graph} is a graph
$G=(K\cup I\cup J,E)$ in which $K$ induces a clique and $G[I\cup J]$ is a
vertex-disjoint collection of stars $K_{1,p}$ with $0\leq p\leq r$. The
set $I$ contains one designated center from each star, and $J$ contains all
its leaves. Each leaf is adjacent only to its own center. Only the centers
have neighbors in $K$, and every center has at least one such neighbor.
Thus $|I|$ is the number of stars, excluding both clique vertices and leaves.
We call this the \emph{non-strict $K_{1,r}$-split graph} class, since its
stars may have different numbers of leaves between zero and $r$. The graph is
a \emph{strict $K_{1,r}$-split graph} if every
component of $G[I\cup J]$ is isomorphic to $K_{1,r}$. Thus strict
$K_{1,0}$-split graphs are precisely split graphs. For $u\in I$, write
$N_G^J(u)=N_G(u)\cap J$ and $d_G^J(u)=|N_G^J(u)|$. In a strict
$K_{1,1}$-split graph, every center has a unique leaf and hence $|I|=|J|$.
We call instances \emph{balanced} when $|K|\leq c|I|$ for a fixed constant
$c$ independent of graph size: the number of clique vertices is at most
$c$ times the number of star centers.

\section{Perfect Italian Domination in Split Graphs}\label{sec:pid-split}
Throughout this section, let $G=(K\cup I,E)$ be a connected split graph with a
fixed split partition, where $K$ is a clique and $I$ is an independent set. For
a function $f:V(G)\rightarrow\{0,1,2\}$ and a set $X\subseteq V(G)$, write
$f(X)=\sum_{x\in X}f(x)$.

\begin{lemma}\label{lem:clique-dichotomy}
For every PID-function $f$ of $G$, either $f(v)>0$ for all $v\in K$, or
$f(K)\leq2$. More precisely, exactly one of the following holds:
\begin{enumerate}[label=\textnormal{(\roman*)}]
    \item every vertex of $K$ has label zero;
    \item exactly one vertex of $K$ has label one and at least one vertex of
    $K$ has label zero;
    \item exactly one vertex of $K$ has label two and at least one vertex of
    $K$ has label zero;
    \item exactly two vertices of $K$ have label one and at least one vertex
    of $K$ has label zero; or
    \item every vertex of $K$ has a positive label.
\end{enumerate}
\end{lemma}

\begin{proof}
Let $v\in K$ satisfy $f(v)=0$. Since $K$ is a clique, every vertex of
$K\setminus\{v\}$ is adjacent to $v$. Hence
$f(K)=\sum_{u\in K\setminus\{v\}}f(u) \leq \sum_{u\in N_G(v)}f(u)=2$.
If $K$ has a zero-labeled vertex, the inequality $f(K)\leq2$ gives precisely
cases (i)--(iv). If it has none, case (v) holds. The alternatives are therefore
exhaustive and pairwise disjoint.
\qed
\end{proof}


\begin{lemma}\label{lem:low-clique}
The first four cases of Lemma~\ref{lem:clique-dichotomy} can be solved in
polynomial time.
\end{lemma}

\begin{proof}
\begin{enumerate}[label=\textnormal{(\roman*)}]
    \item If $f(K)=0$, a valid extension exists if and only if
    $1\leq d_G^I(v)\leq2$ for every $v\in K$ and the requirements
    $d_G^I(v)=1 \Longrightarrow f(x)=2$
    for the unique $x\in N_G^I(v)$, and
    $d_G^I(v)=2 \Longrightarrow f(x)=f(y)=1$
    for $N_G^I(v)=\{x,y\}$ are mutually consistent.

    \item If $f(p)=1$ and $f(K\setminus\{p\})=0$, the minimum extension assigns
    one to every vertex of $I$. It is valid if and only if
    $d_G^I(v)=1$ for all $v\in K\setminus\{p\}$, and has weight $|I|+1$.

    \item If $f(p)=2$ and $f(K\setminus\{p\})=0$, an extension exists if and
    only if $I\subseteq N_G(p)$. In that case all vertices other than $p$ have
    label zero, and the weight is two.

    \item If $f(p)=f(q)=1$ and all other clique vertices have label zero, the
    minimum extension is
    \[
       f(x)=\begin{cases}
       0, & \{p,q\}\subseteq N_G(x),\\
       1, & \text{otherwise},
       \end{cases}
       \qquad x\in I.
    \]
    It is valid if and only if
    $N_G^I(K\setminus\{p,q\})\subseteq N_G(p)\cap N_G(q)$, and has weight
    $2+\bigl|\{x\in I:\{p,q\}\nsubseteq N_G(x)\}\bigr|$.
\end{enumerate}
Suppose first that $f(K)=0$. A vertex $x\in I$ cannot have label zero, since
all its neighbors then have label zero. Thus every contribution from $I$ to a
zero-labeled clique vertex is positive. The required neighbor-sum two implies
that a clique vertex has either one independent neighbor of label two or two
independent neighbors, both of label one. Degrees zero and at least three are
therefore impossible. Since $G$ is connected, every vertex of $I$ occurs in at
least one such requirement. Checking whether the resulting forced labels are
consistent gives both necessity and sufficiency.

Next suppose that $p$ is the unique clique vertex of label one. For every
$x\in I$, the label sum in $N_G(x)\cap K$ is at most one, so $x$ cannot have
label zero. Minimality therefore assigns label one to every vertex of $I$.
Each vertex of $K\setminus\{p\}$ already receives one from $p$ and must receive
exactly one further unit from $I$, which is equivalent to $d_G^I(v)=1$.

Suppose that $p$ is the unique clique vertex of label two. If some $x\in I$
is not adjacent to $p$, then $x$ cannot have label zero and hence has a
positive label in a minimum extension. Since $G$ is connected, $x$ has a
neighbor $v\in K\setminus\{p\}$. But $v$ already receives two from $p$, and
the positive value of $x$ makes its neighbor-sum larger than two, a
contradiction. Conversely, if $p$ is adjacent to all of $I$, assigning zero to
all other vertices is a PID-function.

Finally, suppose that $p$ and $q$ have label one. An independent vertex can
have label zero exactly when it is adjacent to both $p$ and $q$; otherwise it
must be positive and is assigned label one in a minimum extension. Every
vertex of $K\setminus\{p,q\}$ already receives two from $p$ and $q$, so all
its independent neighbors must have label zero. This proves the stated
condition and weight formula.
\qed
\end{proof}

Lemma~\ref{lem:low-clique} resolves all cases in which the clique contains a
zero-labeled vertex. It remains to consider the fifth case of
Lemma~\ref{lem:clique-dichotomy}, where every clique vertex is positive. Once
the labels on $K$ are fixed, the vertices of $I$ can be optimized
independently, leading to the following  extension.

\begin{lemma}\label{lem:canonical-extension}
Let $g:K\rightarrow\{1,2\}$. Its unique minimum-weight extension to $I$ is
\[
 f(x)=\begin{cases}
 0, & \displaystyle\sum_{v\in N_G(x)\cap K}g(v)=2,\\
 1, & \text{otherwise}.
 \end{cases}
\]
\end{lemma}

\begin{proof}
There are no zero-labeled vertices in $K$, so the labels assigned to vertices
of $I$ do not have to satisfy any condition at a clique vertex. A vertex
$x\in I$ can be assigned zero precisely when its clique-neighbor sum is two.
If that sum is not two, $x$ must be positive, and label one is optimal.
\qed
\end{proof}

To complete case (v) of Lemma~\ref{lem:clique-dichotomy}, we must optimize the
 extension of Lemma~\ref{lem:canonical-extension} over all positive
labelings of $K$. For this purpose, construct a weighted graph $H$ with vertex
set $K$. For $v\in K$, let
$p(v)=\bigl|\{x\in I:N_G(x)=\{v\}\}\bigr|$,
and, for distinct $u,v\in K$, let
$q(uv)=\bigl|\{x\in I:N_G(x)=\{u,v\}\}\bigr|$.
Only pairs with positive weight need to be included as edges of $H$. Let
$r=\bigl|\{x\in I:d_G(x)\geq3\}\bigr|$.
For a partition $K=K_1\mathbin{\dot\cup}K_2$, define
\begin{equation}
 \Phi(K_1,K_2)=r+|K_1|+2|K_2|+\sum_{v\in K_1}p(v)
 +\sum_{\substack{uv\in E(H)\\\{u,v\}\cap K_2\neq\emptyset}}q(uv).
 \label{eq:partition-cost}
\end{equation}

\begin{lemma}\label{lem:pid-to-partition}
The fifth case of Lemma~\ref{lem:clique-dichotomy}, in which
$f(v)>0$ for every $v\in K$, has minimum PID weight
$\min\Phi(K_1,K_2)$. An optimal partition labels $K_1$ by one and $K_2$ by
two and is extended to $I$ as in Lemma~\ref{lem:canonical-extension}.
\end{lemma}

\begin{proof}
The clique contributes $|K_1|+2|K_2|$. If $x\in I$ has degree at least
three, its clique-neighbor sum is at least three, so the label of $x$ is fixed to be one. A degree-one vertex adjacent to $v$ contributes one exactly when
$v\in K_1$. A degree-two vertex with neighborhood $\{u,v\}$ contributes one
exactly when at least one of $u,v$ belongs to $K_2$. These are precisely the
terms in the partition-cost formula \eqref{eq:partition-cost}. Conversely, the labeling obtained from a
partition and Lemma~\ref{lem:canonical-extension} has exactly this weight.
\qed
\end{proof}

Figure~\ref{fig:pid-partition} illustrates the weighted graph $H$ associated with a split graph $G$ and an optimal partition $(K_1,K_2)$ from
Lemma~\ref{lem:pid-to-partition}.
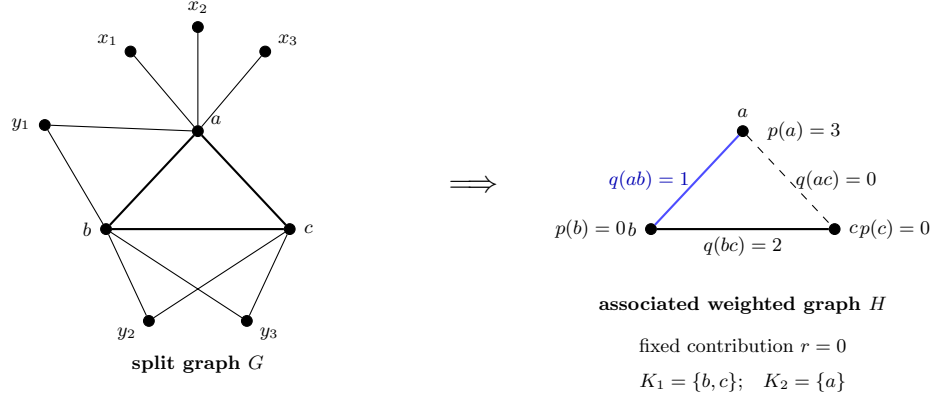
\begin{figure}[t]
\centering
\begin{tikzpicture}[scale=0.81,transform shape,
  every node/.style={font=\small}]
  \begin{scope}[xshift=-4.5cm]
    \node[vertexdot] (a) at (0,1.6) {};
    \node[vertexdot] (b) at (-1.5,0) {};
    \node[vertexdot] (c) at (1.5,0) {};
    \node[vertexlabel,above right=0mm and 1mm of a] {$a$};
    \node[vertexlabel,left=1mm of b] {$b$};
    \node[vertexlabel,right=1mm of c] {$c$};
    \draw[thick] (a)--(b)--(c)--(a);
    \node[vertexdot] (pa1) at (-1.1,2.9) {};
    \node[vertexdot] (pa2) at (0,3.3) {};
    \node[vertexdot] (pa3) at (1.1,2.9) {};
    \node[vertexdot] (ab) at (-2.5,1.7) {};
    \node[vertexdot] (bc1) at (-0.8,-1.5) {};
    \node[vertexdot] (bc2) at (0.8,-1.5) {};
    \node[vertexlabel,above left=0mm and 1mm of pa1] {$x_1$};
    \node[vertexlabel,above=1mm of pa2] {$x_2$};
    \node[vertexlabel,above right=0mm and 1mm of pa3] {$x_3$};
    \node[vertexlabel,left=1mm of ab] {$y_1$};
    \node[vertexlabel,below left=0mm and 1mm of bc1] {$y_2$};
    \node[vertexlabel,below right=0mm and 1mm of bc2] {$y_3$};
    \draw (a)--(pa1) (a)--(pa2) (a)--(pa3) (a)--(ab) (b)--(ab);
    \draw (b)--(bc1) (c)--(bc1) (b)--(bc2) (c)--(bc2);
    \node[font=\bfseries] at (0,-2.2) {split graph $G$};
  \end{scope}
  \node[font=\Large] at (0,0.7) {$\Longrightarrow$};
  \begin{scope}[xshift=4.4cm]
    \node[vertexdot] (ha) at (0,1.6) {};
    \node[vertexdot] (hb) at (-1.5,0) {};
    \node[vertexdot] (hc) at (1.5,0) {};
    \node[vertexlabel,above=1mm of ha] {$a$};
    \node[vertexlabel,left=1mm of hb] {$b$};
    \node[vertexlabel,right=1mm of hc] {$c$};
    \draw[blue!70,thick] (ha)--node[left,text=blue!70!black]
      {$q(ab)=1$} (hb);
    \draw[thick] (hb)--node[below] {$q(bc)=2$} (hc);
    \draw[dashed] (ha)--node[right] {$q(ac)=0$} (hc);
    \node[right=2mm of ha] {$p(a)=3$};
    \node[left=2mm of hb] {$p(b)=0$};
    \node[right=2mm of hc] {$p(c)=0$};
    \node[font=\bfseries] at (0,-1.25) {associated weighted graph $H$};
    \node at (0,-1.9) {fixed contribution $r=0$};
    \node at (0,-2.5) {$K_1=\{b,c\}$;\quad $K_2=\{a\}$};
  \end{scope}
\end{tikzpicture}
\caption{A split graph $G$ with maximal clique $K=\{a,b,c\}$ and its weighted
graph $H$. Degree-one vertices of $G$ become vertex weights in $H$, degree-two
vertices become edge weights, and vertices of degree at least three contribute
the constant $r$; here $r=0$. The blue edge is the edge-weight contribution
selected by the illustrated optimal partition $K_1=\{b,c\}$,
$K_2=\{a\}$.}
\label{fig:pid-partition}
\end{figure}
We now complete the treatment of case (v) of
Lemma~\ref{lem:clique-dichotomy} by reducing the weighted partition instances
of Lemma~\ref{lem:pid-to-partition} to a minimum-capacity cut problem.

\subsection{Minimum-Capacity Cut Formulation}\label{sec:partition-cut}
The \emph{minimum-capacity $s$--$t$ cut problem} takes as input a directed
network $D=(V_D,A_D)$, a source $s$, a sink $t$, and a nonnegative capacity
$c(a)$ for every arc $a\in A_D$. An $s$--$t$ cut is a partition $(S,T)$ of
$V_D$ with $s\in S$ and $t\in T$, and its capacity is
$c(S,T)=\sum_{xy\in A_D:\,x\in S,\,y\in T}c(xy)$.
The objective is to find an $s$--$t$ cut of minimum capacity.

Consider a weighted partition instance $(H,p,q,r)$ obtained in
Lemma~\ref{lem:pid-to-partition}, where $V(H)=K$. Map this instance to a
network $D(H,p,q)$ as follows. For each $v\in K$, create a vertex node $v'$,
and for each $uv\in E(H)$, create a pair node $e_{uv}$. Add the source $s$,
the sink $t$, and set
$M=1+|K|+\sum_{v\in K}p(v)+\sum_{uv\in E(H)}q(uv)$.
Add arcs $s\to v'$ of capacity one, $v'\to t$ of capacity $p(v)$,
$s\to e_{uv}$ of capacity $q(uv)$, and $e_{uv}\to u'$ and
$e_{uv}\to v'$ of capacity $M$. This construction has polynomial size. A
partition $K=K_1\mathbin{\dot\cup}K_2$ is mapped to the cut whose source side
is
\[
 S(K_1)=\{s\}\cup\{v':v\in K_1\}
       \cup\{e_{uv}:u,v\in K_1\},
\]
and whose sink side is $V_D\setminus S(K_1)$.

\begin{lemma}
\label{lem:partition-to-cut}
Let $(H,p,q,r)$ be a weighted partition instance with objective function
$\Phi$ defined in \eqref{eq:partition-cost}, and let $D(H,p,q)$ be the network
constructed above. Then
\[
 \min_{K=K_1\mathbin{\dot\cup}K_2}\Phi(K_1,K_2)
 =r+|K|+\lambda(D),
\]
where $\lambda(D)$ is the minimum $s$--$t$ cut capacity in $D$.
\end{lemma}

\begin{proof}
Putting all vertex and pair nodes on the sink side gives a cut of capacity at
most $|K|+\sum_{uv\in E(H)}q(uv)<M$. Thus no minimum cut contains a
capacity-$M$ arc.

For a partition $K=K_1\mathbin{\dot\cup}K_2$, the mapped cut
$(S(K_1),V_D\setminus S(K_1))$ has capacity
\begin{equation}
 \Psi(K_1,K_2)=|K_2|+\sum_{v\in K_1}p(v)
 +\sum_{\substack{uv\in E(H)\\\{u,v\}\cap K_2\neq\emptyset}}q(uv).
 \label{eq:cut-residual}
\end{equation}

Conversely, let $(S,T)$ be a minimum cut. Since it contains no capacity-$M$
arc, $e_{uv}\in S$ implies $u',v'\in S$. If $u',v'\in S$ but
$e_{uv}\in T$, moving $e_{uv}$ to $S$ decreases the cut capacity by
$q(uv)>0$, a contradiction. Therefore every minimum cut satisfies
$e_{uv}\in S$ exactly when
$u',v'\in S$. With
 $K_1=\{v\in K:v'\in S\},\qquad K_2=K\setminus K_1$,
the cut capacity is exactly \eqref{eq:cut-residual}. Since
$\Phi(K_1,K_2)=r+|K|+\Psi(K_1,K_2)$, minimizing the cut capacity is equivalent
to minimizing the partition cost.
\qed
\end{proof}

Consequently, the mapping $K_1\mapsto S(K_1)$ sends every optimal partition
to a minimum cut. Conversely, a minimum cut $(S,T)$ yields an optimal
partition by setting $K_1=\{v\in K:v'\in S\}$ and
$K_2=K\setminus K_1$. Combining Lemma~\ref{lem:partition-to-cut} with the
five cases of Lemma~\ref{lem:clique-dichotomy} proves
Theorem~\ref{thm:pid-split}.
Figure~\ref{fig:partition-cut} illustrates the minimum-cut network associated
with the weighted graph in Figure~\ref{fig:pid-partition} and the cut
corresponding to its optimal partition.

\begin{figure}[t]
\centering
\begin{tikzpicture}[scale=0.92,transform shape,>=stealth,
  every node/.style={font=\small}]
  \node[vertexdot] (s) at (-5.2,0) {};
  \node[vertexdot] (eab) at (-2.5,1.25) {};
  \node[vertexdot] (ebc) at (-2.5,-1.25) {};
  \node[vertexdot] (a) at (0,2.35) {};
  \node[vertexdot] (b) at (0,0) {};
  \node[vertexdot] (c) at (0,-2.35) {};
  \node[vertexdot] (t) at (4.6,0) {};
  \node[vertexlabel,left=1mm of s] {$s$};
  \node[vertexlabel,above=1mm of eab] {$e_{ab}$};
  \node[vertexlabel,below=1mm of ebc] {$e_{bc}$};
  \node[vertexlabel,above=1mm of a] {$a'$};
  \node[vertexlabel,right=1mm of b] {$b'$};
  \node[vertexlabel,below=1mm of c] {$c'$};
  \node[vertexlabel,right=1mm of t] {$t$};

  \draw[->,blue!70,thick] (s)--node[above left] {$q(ab)=1$} (eab);
  \draw[->,thick] (s)--node[below left] {$q(bc)=2$} (ebc);
  \draw[->,blue!70,thick] (s) to[bend left=42] node[above] {$1$} (a);
  \draw[->] (s)--node[above] {$1$} (b);
  \draw[->] (s) to[bend right=42] node[below] {$1$} (c);

  \draw[->,very thick] (eab)--node[above] {$M$} (a);
  \draw[->,very thick] (eab)--node[pos=0.35,below] {$M$} (b);
  \draw[->,very thick] (ebc)--node[pos=0.35,above] {$M$} (b);
  \draw[->,very thick] (ebc)--node[below] {$M$} (c);

  \draw[->] (a) to[bend left=24]
    node[above] {$p(a)=3$} (t);
  \draw[->,blue!70,thick] (b)--node[above] {$p(b)=0$} (t);
  \draw[->,blue!70,thick] (c) to[bend right=24]
    node[below] {$p(c)=0$} (t);

  \node[draw,dashed,rounded corners,fit=(eab)(ebc),inner sep=5mm] {};
  \node[draw,dashed,rounded corners,fit=(a)(b)(c),inner sep=5mm] {};
  \node[align=center] at (0,-3.55) {Source-side vertex nodes correspond to
  $K_1$; Sink-side vertex nodes correspond to $K_2$};
\end{tikzpicture}
\caption{Minimum-cut network for Figure~\ref{fig:pid-partition}. The dashed
regions contain the edge nodes and vertex nodes of $H$, respectively; very
thick arcs have capacity $M$. The blue arcs form a minimum-capacity cut of
value two, corresponding to $K_1=\{b,c\}$ and $K_2=\{a\}$. In particular,
the cut contains the source-to-edge-node arc $s\to e_{ab}$.}
\label{fig:partition-cut}
\end{figure}

\begin{theorem}\label{thm:pid-split}
PID is polynomial-time solvable on split
graphs. Moreover, a minimum-weight PID-function can be computed in
$O(n^4)$ time, where $n=|V(G)|$.
\end{theorem}

\begin{proof}
By Lemma~\ref{lem:clique-dichotomy}, every PID-function belongs to one of five
exhaustive clique-labeling cases. Lemma~\ref{lem:low-clique} computes the
minimum valid function in the first four cases. For the remaining case,
Lemma~\ref{lem:pid-to-partition} reduces the problem to minimizing
\eqref{eq:partition-cost}, and Lemma~\ref{lem:partition-to-cut} reduces this
partition problem to a minimum-capacity $s$--$t$ cut. By the Max-Flow Min-Cut
Theorem of Ford and Fulkerson~\cite{ford1956maxflow}, a minimum $s$--$t$ cut
in the constructed network $D(H,p,q)$ can be obtained from a maximum $s$--$t$
flow in the same network. The Edmonds--Karp algorithm computes this maximum
$s$--$t$ flow in $O(|V(D)|\,|E(D)|^2)$ time~\cite{edmondskarp1972}.

Let $n=|V(G)|$, $m=|E(G)|$, and $k=|K|$. There are
$1+k+k+\binom{k}{2}=O(k^2)$ low clique-weight candidates, each constructible
and verifiable in $O(n+m)$ time, for a total of $O(k^2(n+m))$. The weighted
graph $H$ is constructed in $O(n+m+k^2)$ time. Every edge of $H$ has positive
weight and represents the neighborhood of at least one degree-two vertex of
$I$. Distinct edges represent distinct neighborhoods, so
 $|E(H)|\leq |\{x\in I:d_G(x)=2\}|\leq |I|$.

Thus $D$ has $2+k+|E(H)|=O(n)$ vertices and $2k+3|E(H)|=O(n)$ arcs.
Edmonds--Karp runs in $O(n^3)$ time on this network. All capacities are
nonnegative integers bounded by $n+1$, so they have $O(\log n)$-bit
representations. The complete algorithm therefore runs in  $O\bigl(k^2(n+m)+n^3\bigr)= O(n^4)$,
which proves the theorem.
\qed
\end{proof}

\section{Hardness at Split-Deletion Distance One}
\label{sec:split-plus-one}
For a graph $G$, define its deletion distance to the class of split graphs by
\[
  \operatorname{sd}(G)
  =\min\bigl\{|S|:S\subseteq V(G)\text{ and }G-S
      \text{ is a split graph}\bigr\}.
\]
\begin{theorem}
\label{thm:split-plus-one-hardness}
\textnormal{\textsc{PID}} is NP-complete on graphs $G$ with
$\operatorname{sd}(G)=1$.
\end{theorem}

\begin{proof}
Membership in NP is immediate.  For NP-hardness, we reduce from
\textsc{Restricted Exact Cover by 3-Sets} (RXC3), which is NP-complete
\cite{gonzalez1985clustering}.

Let $(U,\mathcal F)$ be an RX3C instance, where
$U=\{u_1,\ldots,u_{3q}\}$
and $\mathcal F$ consists of $3q$ three-element subsets of $U$, with every
element occurring in exactly three members of $\mathcal F$.  The question is
whether $\mathcal F$ contains a subfamily that partitions $U$.  We may assume
that $q\geq2$, since the finitely many smaller instances can be decided
directly and mapped to fixed yes- or no-instances.

Set
$b=q+2$ and $L=b+1=q+3$.
We construct a graph $G$ with the following vertices:
\begin{itemize}
  \item one distinguished vertex $s$;
  \item an element vertex $k_u$ for every $u\in U$;
  \item one vertex $w$;
  \item test vertices $T=\{t_1,\ldots,t_L\}$; and
  \item a set vertex $x_F$ for every $F\in\mathcal F$.
\end{itemize}
Put $ K_U=\{k_u:u\in U\},  C=\{w\}\cup K_U\cup T, X=\{x_F:F\in\mathcal F\}$.
Make $C$ a clique and $X$ an independent set.  Join $s$ to every vertex of
$T\cup X$, join $w$ to every vertex of $X$, and join $x_F$ to $k_u$
exactly when $u\in F$.  No other edges are added.  In particular, $s$ is
not adjacent to $w$ or to an element vertex, and the test vertices have no
neighbors in $X$. Figure~\ref{fig:split-plus-one-full} shows the complete construction for an
illustrative RXC3 instance with six elements and six sets.

\begin{figure}[t]
\centering
\begin{tikzpicture}[scale=0.70,transform shape,
  every node/.style={font=\scriptsize}]
  \foreach \n/\lab/\val/\yy in {
    t1/t_1/0/4.0,t2/t_2/0/3.45,t3/t_3/0/2.9,
    t4/t_4/0/2.35,t5/t_5/0/1.8,w/w/1/1.05,
    k1/k_{u_1}/0/0.25,k2/k_{u_2}/0/-0.35,
    k3/k_{u_3}/0/-0.95,k4/k_{u_4}/0/-1.55,
    k5/k_{u_5}/0/-2.15,k6/k_{u_6}/0/-2.75}
  {
    \node[vertexdot] (C\n) at (-4.9,\yy+0.2) {};
    \node[vertexlabel,left=1mm of C\n] {$\lab$};
    \node[pidvalue,left=10mm of C\n] {$\val$};
  }
  \node[draw,dashed,rounded corners,
    fit=(Ct1)(Ct2)(Ct3)(Ct4)(Ct5)(Cw)(Ck1)(Ck2)(Ck3)(Ck4)(Ck5)(Ck6),
    inner sep=6mm,label=above:{ }] {};

  \foreach \n/\lab/\val/\yy in {
    x1/x_{F_1}/0/2.8,x2/x_{F_2}/1/1.8,x3/x_{F_3}/0/0.8,
    x4/x_{F_4}/1/-0.2,x5/x_{F_5}/0/-1.2,x6/x_{F_6}/0/-2.2}
  {
    \node[vertexdot] (X\n) at (0,\yy) {};
    \node[vertexlabel,above=0.5mm of X\n] {$\lab$};
    \node[pidvalue,below=0.5mm of X\n] {$\val$};
  }

  \node[vertexdot] (Es) at (4.2,0.3) {};
  \node[vertexlabel,right=1mm of Es] {$s$};
  \node[pidvalue,right=7mm of Es] {$1$};

  \draw (Ck1)--(Xx1) (Ck2)--(Xx1) (Ck3)--(Xx1);
  \draw (Ck1)--(Xx2) (Ck2)--(Xx2) (Ck4)--(Xx2);
  \draw (Ck1)--(Xx3) (Ck5)--(Xx3) (Ck6)--(Xx3);
  \draw (Ck3)--(Xx4) (Ck5)--(Xx4) (Ck6)--(Xx4);
  \draw (Ck2)--(Xx5) (Ck4)--(Xx5) (Ck5)--(Xx5);
  \draw (Ck3)--(Xx6) (Ck4)--(Xx6) (Ck6)--(Xx6);

  \foreach \x in {Xx1,Xx2,Xx3,Xx4,Xx5,Xx6}
    \draw[blue!65] (Cw)--(\x);

  \foreach \x in {Xx1,Xx2,Xx3,Xx4,Xx5,Xx6}
    \draw[blue!65] (Es)--(\x);
  \draw[blue!65] (Es) .. controls (3.5,4.9) and (-3.5,4.9) .. (Ct1);
  \draw[blue!65] (Es) .. controls (3.3,4.6) and (-3.4,4.6) .. (Ct2);
  \draw[blue!65] (Es) .. controls (3.1,4.3) and (-3.3,4.3) .. (Ct3);
  \draw[blue!65] (Es) .. controls (2.9,4.0) and (-3.2,4.0) .. (Ct4);
  \draw[blue!65] (Es) .. controls (2.7,3.7) and (-3.1,3.7) .. (Ct5);

  \node[font=\bfseries] at (-4.9,4.85) {$C$};
  \node[font=\bfseries] at (0,4.65) {$X$};
  \node[font=\bfseries] at (4.2,1.05) { };
  \node[align=center] at (0,-3.55)
    {Colored edges are just used for better visualization.};
\end{tikzpicture}
\caption{The graph of split-deletion distance one used in the hardness reduction for an
illustrative RXC3 instance, with PID values shown in red. The displayed
PID-function corresponds to the exact cover $\{F_2,F_4\}$. Clique edges inside
$C$ are suppressed. The black edges encode set membership, the blue edges join
$w$ to every set vertex and join $s$ to all test and set vertices. Deleting
$s$ leaves the split partition $(C,X)$.}
\label{fig:split-plus-one-full}
\end{figure}
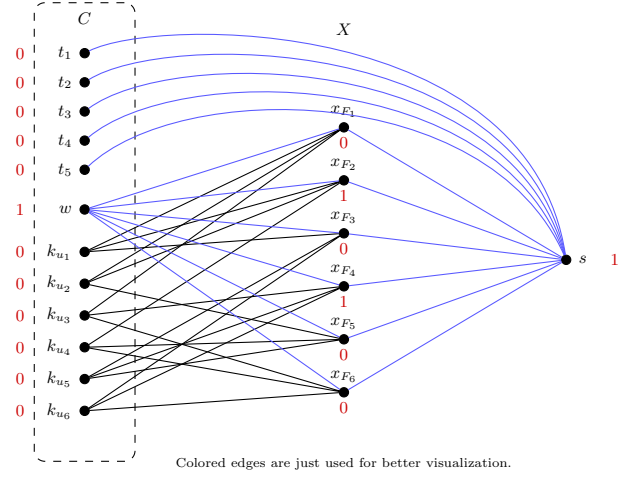

The graph $G-s$ has the split partition $(C,X)$. For any $F\in\mathcal F$,
$u\in F$, and $t\in T$, the vertices $s,t,k_u,x_F$ induce a four-cycle in
that order. Hence $G$ is not split, and therefore
$\operatorname{sd}(G)=1$.
We claim that
\begin{equation}
  (U,\mathcal F)\text{ has an exact cover}
  \quad\Longleftrightarrow\quad
  \gamma_I^P(G)\leq b.
  \label{eq:split-plus-one-equivalence}
\end{equation}

Suppose first that $\mathcal F'\subseteq\mathcal F$ is an exact cover of
$U$.  Necessarily $|\mathcal F'|=q$.  Define a function $f$ by
$f(s)=f(w)=1$,
\[
  f(x_F)=
  \begin{cases}
    1, & F\in\mathcal F',\\
    0, & F\notin\mathcal F',
  \end{cases}
\]
and assign label zero to every vertex of $K_U\cup T$.  A test vertex sees
one unit from $s$ and one from $w$.  An element vertex $k_u$ sees one unit
from $w$ and one from the unique selected set containing $u$.  Finally, an
unselected set vertex sees one unit from each of $s$ and $w$.  Hence $f$ is
a PID-function, and
$\omega(f)=2+|\mathcal F'|=q+2=b$.

Conversely, suppose that $G$ has a PID-function $f$ with
$\omega(f)\leq b$.  Define
$c=f(C)$ and $x=f(X)=\sum_{F\in\mathcal F}f(x_F)$.
Since $|T|=L>b$, at least one test vertex $t\in T$ has label zero.  Its
neighbors are $s$ and all vertices of $C\setminus\{t\}$.  Since
$f(t)=0$, the PID condition at $t$ gives
\begin{equation}
  f(s)+c=2.
  \label{eq:split-plus-one-forcing}
\end{equation}
We distinguish the three possible labels of $s$.

\smallskip
\noindent\emph{Case 1: $f(s)=0$.}
Equation~\eqref{eq:split-plus-one-forcing} gives $c=2$.  The PID condition
at $s$, whose neighborhood is $T\cup X$, gives
\begin{equation}
  f(T)+x=2.
  \label{eq:split-plus-one-s-zero}
\end{equation}
At most two element vertices are positive because their labels are included
in the total clique weight $c=2$.  Hence every set
$F\in\mathcal F$ contains an element $u$ for which $f(k_u)=0$.  The
neighbors of $k_u$ inside $C$ have total label $c=2$, so its PID condition
forces
$\sum_{F'\in\mathcal F:\,u\in F'}f(x_{F'})=0$.
In particular, $f(x_F)=0$.  Since this holds for every $F$, we obtain
$x=0$.

It now follows from \eqref{eq:split-plus-one-s-zero} that $f(T)=2$.
Because $f(C)=2$, the vertices $w$ and all element vertices have label
zero.  Every set vertex is also zero, but all its neighbors---namely $s$,
$w$, and its three element vertices---are zero.  This contradicts the PID
condition at the set vertex.  Therefore this case is impossible.

\smallskip
\noindent\emph{Case 2: $f(s)=2$.}
Equation~\eqref{eq:split-plus-one-forcing} gives $c=0$.  Thus every vertex
of $C$ is zero.  For each $u\in U$, the PID condition at $k_u$ becomes
$\sum_{F\in\mathcal F:\,u\in F}f(x_F)=2$.
Summing over $U$ and using $|F|=3$ for every $F\in\mathcal F$ yields
\[
  6q
  =\sum_{u\in U}\sum_{\substack{F\in\mathcal F\\u\in F}}f(x_F)
  =3\sum_{F\in\mathcal F}f(x_F)
  =3x.
\]
Consequently $x=2q$, and therefore
$\omega(f)\geq f(s)+x=2+2q>q+2=b$,
contrary to the assumed weight bound.

\smallskip
\noindent\emph{Case 3: $f(s)=1$.}
Equation~\eqref{eq:split-plus-one-forcing} gives $c=1$.  Hence exactly one
vertex $h\in C$ has label one, while every other vertex of $C$ has label
zero.

We first show that $h$ is not an element vertex.  Suppose that
$h=k_y$ for some $y\in U$.  Then $w$ is zero.  Since $w$ is adjacent to
every other vertex of $C$ and to every vertex of $X$, its PID condition
gives
$2=c+x=1+x$,
so $x=1$.  For every $u\in U\setminus\{y\}$, the zero-labeled vertex
$k_u$ receives one unit from $C$ and must receive one further unit from its
set neighbors.  Hence
$\sum_{F\in\mathcal F:\,u\in F}f(x_F)=1$.
Because $x=1$, precisely one set vertex has label one.  It is adjacent to
only three element vertices and therefore cannot supply one unit to all
$3q-1\geq5$ vertices in $U\setminus\{y\}$.  This is a contradiction.

Thus $h\in\{w\}\cup T$, and every element vertex has label zero.  For each
$u\in U$, the PID condition at $k_u$ now gives
\begin{equation}
  2=c+\sum_{\substack{F\in\mathcal F\\u\in F}}f(x_F)
   =1+\sum_{\substack{F\in\mathcal F\\u\in F}}f(x_F).
  \label{eq:split-plus-one-exact-cover}
\end{equation}
Therefore
$\sum_{F\in\mathcal F:\,u\in F}f(x_F)=1$ for every $u\in U$.
Let
$\mathcal F'=\{F\in\mathcal F:f(x_F)>0\}$.
Equation~\eqref{eq:split-plus-one-exact-cover} implies that every positive
set vertex has label one and that every element belongs to exactly one member
of $\mathcal F'$.  Hence $\mathcal F'$ is an exact cover of $U$.

The three cases prove the reverse implication in
\eqref{eq:split-plus-one-equivalence}.  The construction has polynomial size,
which completes the reduction.
\qed
\end{proof}

The following corollary \ref{cor:split-distance-para-np} is a direct consequence of Theorem \ref{thm:split-plus-one-hardness}.
\begin{corollary}
\label{cor:split-distance-para-np}
\textnormal{\textsc{PID}} parameterized by deletion distance to split graphs is para-NP-hard unless $\mathrm{P}=\mathrm{NP}$.
\end{corollary}

\section{Complexity on $K_{1,r}$-Split Graphs}
\label{sec:k11-hardness}

In this section we consider a different
extension of split graphs obtained by attaching private leaves to independent-side
vertices. We prove hardness for strict $K_{1,1}$-split graphs and then
classify the remaining strict and non-strict star classes. Figure~\ref{fig:k11-reduction} illustrates the two
gadgets used in the reduction. Throughout this section, for a function $f$ on
a graph $G=(K\cup I\cup J,E)$, write
$f_K=\sum_{v\in K}f(v)$.
\deferredproofnote
\begin{firstdeferredlemma}\label{lem:strict-recognition}
For every fixed $r\geq1$, strict $K_{1,r}$-split graphs can be recognized
and a defining partition can be recovered in $O(n+m)$ time, where
$n=|V(G)|$ and $m=|E(G)|$. All defining partitions of a connected strict
instance have the same values of $|K|$, $|I|$, and $|J|$.
\end{firstdeferredlemma}

\begin{deferredlemma}\label{lem:no-two-clique}
Let $G=(K\cup I\cup J,E)$ be a strict $K_{1,1}$-split graph with
$I\neq\emptyset$, and let $f$ be a minimum PID-function of $G$. Then $f(v)\leq1$ for every $v\in K$.
\end{deferredlemma}

We reduce from Restricted Exact Cover by 3-Sets (RXC3). An instance consists
of a set $X$ with $|X|=3q$ and a family
$\mathcal C=\{S_1,\ldots,S_{3q}\}$ of three-element subsets of $X$ such that
every element occurs in exactly three sets. The question is whether
$\mathcal C$ contains an exact cover of $X$. RXC3 is NP-complete
\cite{gonzalez1985clustering}. We may assume $q\geq3$, since excluding the
finitely many smaller instances does not affect NP-hardness.

\begin{deferredtheorem}\label{thm:k11-hard}
PID is NP-complete on strict $K_{1,1}$-split graphs.
\end{deferredtheorem}

As a direct consequence of Theorem \ref{thm:k11-hard} we have the complete
classification for the non-strict $K_{1,r}$-split.

\begin{corollary}\label{thm:kr-hard}
For every fixed integer $r\geq1$, PID is NP-complete on non-strict $K_{1,r}$-split
graphs.
\end{corollary}

\begin{deferredtheorem}\label{thm:strict-stars}
On strict $K_{1,r}$-split graphs, PID is NP-complete for $r=\{1,2\}$.
For every fixed $r\geq3$, a minimum PID-function can be computed in
$O(n+m)$ time, where $n=|V(G)|$ and $m=|E(G)|$, without a balance
assumption.
\end{deferredtheorem}

This does not contradict Theorem~\ref{thm:kr-hard}, since the non-strict
classes still allow stars with only one or two leaves.

\section{Approximation and Inapproximability}
\label{sec:approximation}

We will study approximation on strict $K_{1,r}$-split graphs later in this
section. First, we return to graphs at deletion distance one from split
graphs, that is, graphs $G$ with $\operatorname{sd}(G)=1$. By adding more
test vertices to the construction of
Theorem~\ref{thm:split-plus-one-hardness}, we show that even a
$(2-\varepsilon)$-approximation is not possible in polynomial time for any
fixed $0<\varepsilon<1$, unless $\mathrm{P}=\mathrm{NP}$.

\begin{deferredtheorem}\label{thm:split-plus-one-inapprox}
For every fixed $0<\varepsilon<1$, minimum PID has no polynomial-time $(2-\varepsilon)$-approximation algorithm on graphs $G$
with $\operatorname{sd}(G)=1$, unless $\mathrm{P}=\mathrm{NP}$.
\end{deferredtheorem}

Since the class with $\operatorname{sd}(G)=1$ is contained in the class
with $\operatorname{sd}(G)\leq d$ for every fixed $d\geq1$,
Theorem~\ref{thm:split-plus-one-inapprox} gives the same approximation lower
bound on each such class. In particular, none admits a PTAS unless
$\mathrm{P}=\mathrm{NP}$. This lower bound does not exclude a
$2$-approximation or a larger constant-factor approximation.

The strict cases $r=1,2$ are NP-complete, whereas strict $r\geq3$ admits
an exact algorithm by Theorem~\ref{thm:strict-stars}. We prove that strict
$r=1$ has no PTAS, and then give balanced approximation guarantees for
$r=1,2$. The no-PTAS reduction uses bounded three-dimensional matching. An instance
consists of three pairwise disjoint sets $X,Y,Z$, each of cardinality $q$, and
a family $\mathcal T\subseteq X\times Y\times Z$ of $m$ triples. Every element
of $U=X\mathbin{\dot\cup}Y\mathbin{\dot\cup}Z$ occurs in at least one and at
most $B$ triples, where $B$ is a fixed constant. The objective is to maximize
the cardinality $\nu(\mathcal T)$ of a subfamily of pairwise disjoint triples.
Maximum Bounded Three-Dimensional Matching is MAX SNP-complete
\cite{kann1991}. For the approximation reduction we use the stronger
perfect-matching gap of Petrank~\cite[Theorem~4.4]{petrank1994gap}: for a
fixed $B\geq3$
and some rational constant $0<\delta<1$, it is NP-hard to distinguish instances with
$\nu(\mathcal T)=q$ from instances with
$\nu(\mathcal T)\leq(1-\delta)q$. The hard instances have no isolated
elements and satisfy the above occurrence bound. Counting incidences gives
\begin{equation}
 q\leq m\leq Bq. \label{eq:bd3dm-size}
\end{equation}

\begin{deferredtheorem}\label{thm:no-ptas}
Minimum PID on strict $K_{1,1}$-split graphs does not admit a PTAS unless $\mathrm{P}=\mathrm{NP}$. This holds even when $|K|\leq|I|$.
\end{deferredtheorem}

A direct consequence of Theorem \ref{thm:no-ptas} rules out a PTAS
on each of these non-strict classes, even when $|K|\leq|I|$.
For strict graphs, Theorem~\ref{thm:strict-stars} solves $r\geq3$ exactly. Under the balance
assumption these cases admit the following simple guarantees.
\begin{deferredproposition}\label{prop:balanced-approx}
Let $G$ be a strict $K_{1,r}$-split graph with $r\in\{1,2\}$ admitting a
partition $(K,I,J)$ with $|K|\leq c|I|$ for a fixed constant $c$.
There is a linear-time $(c+2)$-approximation for $r=1$ and a linear-time
$(1+c/2)$-approximation for $r=2$.
\end{deferredproposition}

Under $|K|\leq|I|$, Proposition~\ref{prop:balanced-approx} gives ratio three for strict $r=1$ and ratio $3/2$ for strict $r=2$.
Theorem~\ref{thm:no-ptas} rules out a PTAS for the former unless
$\mathrm{P}=\mathrm{NP}$. For strict $r\geq3$, Theorem~\ref{thm:strict-stars} gives ratio one without balance.
\section*{Declaration of AI Use}
An OpenAI AI assistant (Prism) was used during manuscript preparation to assist with language editing, LaTeX formatting, idea formulation, organization, and checking of mathematical arguments. The authors have verified the correctness of the results and take full responsibility for the integrity of the manuscript.
\begingroup
\hbadness=2000

\endgroup
\clearpage
\appendix
\renewcommand{\theHsection}{\Alph{section}}
\section{Appendix}\label{app:deferred-proofs}

\subsection{Proof of Lemma~\ref{lem:strict-recognition}}\label{app:strict-recognition}
\noindent\textbf{Lemma~\ref{lem:strict-recognition}.}
\textit{For every fixed $r\geq1$, strict $K_{1,r}$-split graphs can be recognized
and a defining partition can be recovered in $O(n+m)$ time, where
$n=|V(G)|$ and $m=|E(G)|$. All defining partitions of a connected strict
instance have the same values of $|K|$, $|I|$, and $|J|$.}

\begin{proof}
We consider nonempty connected graphs, as stipulated in the preliminaries.
A proposed partition can be checked in $O(n+m)$ time by scanning adjacency
lists: check that the parts are disjoint and exhaustive, that $K$ is a
nonempty clique, that $I$ is independent, that every center has exactly $r$
neighbors in $J$ and at least one in $K$, and that each vertex of $J$ has
degree one with its neighbor in $I$. The clique test counts, for each vertex
of $K$, its neighbors in $K$. Thus it suffices to produce a constant number
of candidate partitions.

First handle complete graphs with $K=V(G)$ and $I=J=\emptyset$. Also handle
$G=K_{1,r+1}$: choose one leaf as the sole clique vertex, its neighbor as the
sole center, and the other $r$ leaves as $J$. Both cases can be detected in
linear time. For the remaining graphs let $L$ be the set of degree-one
vertices and let $\ell(v)=|N(v)\cap L|$; compute these quantities in one scan.

In any valid partition every center has degree at least $r+1\geq2$, so every
vertex of $L$ belongs to $J$ or $K$. If $|K|\geq3$, no clique vertex has
degree one. If $|K|=1$, a degree-one clique vertex implies $|I|=1$ and gives
the single-star case already handled. If $|K|=2$ and $I\neq\emptyset$, at
most one clique vertex $b$ has degree one. In that event $b$ has no center
neighbors, and every center is adjacent to the other clique vertex $a$.
These exhaust the possible degree-one clique vertices.

For $r\geq2$, outside the two initial cases every true center has exactly
$r$ neighbors in $L$, while each clique vertex has at most one. Hence the
only possible center set is
$C=\{v:\ell(v)\geq r\}$.
Take $I=C$, $J=L\cap N(C)$, and $K=V(G)\setminus(I\cup J)$, and verify this
candidate. It recovers every valid instance; if verification fails, reject.

For $r=1$, put $C=N(L)$ and first test the candidate
$(K,I,J)=(V(G)\setminus(C\cup L),C,L)$. It succeeds whenever no clique
vertex has degree one. Otherwise the only possible valid form has
$K=\{a,b\}$ as described above. Then $C=I\cup\{a\}$,
$L=J\cup\{b\}$, every vertex of $C$ has exactly one neighbor in $L$, and
$G[C]$ is the star with center $a$ and leaf set $I$. Test whether $G[C]$ is
a star with at least one edge. If it is, there is one possible center $a$,
or two choices when $G[C]=K_2$. For each choice having a unique neighbor
$b$ in $L$, test
$(K,I,J)=(\{a,b\},C\setminus\{a\},L\setminus\{b\})$.
Reject if no candidate passes verification. The exceptional form above shows
that every valid instance is found. There are at most three verification
passes, and all candidate construction and star tests use $O(n+m)$ time.

Finally, complete graphs have only the clique-only partition sizes, and
$K_{1,r+1}$ has $|K|=|I|=1$ and $|J|=r$. Outside these cases, for $r\geq2$
the sets are forced by the construction. For $r=1$, either $L=J$ and $C=I$
in every valid partition, or the exceptional form forces $|K|=2$,
$|I|=|C|-1$, and $|J|=|L|-1$. The possible choice of an endpoint when
$G[C]=K_2$ changes vertices but not part sizes. Thus the sizes are invariant.
\qed
\end{proof}

\subsection{Proof of Lemma~\ref{lem:no-two-clique}}\label{app:no-two-clique}
\noindent\textbf{Lemma~\ref{lem:no-two-clique}.}
\textit{Let $G=(K\cup I\cup J,E)$ be a strict $K_{1,1}$-split graph with
$I\neq\emptyset$, and let $f$ be a minimum PID-function of $G$. Then $f(v)\leq1$ for every $v\in K$.}

\begin{proof}
The assumption $I\neq\emptyset$ excludes clique-only instances. For each
$i\in I$, denote its unique neighbor in $J$ by $i'$. Suppose that
$f(w)=2$ for some $w\in K$.

Assume first that every vertex of $K$ is positive. No center in
$N_G(w)\cap I$ can have label zero: if $f(i)=0$, then the contribution two
from $w$ forces $f(i')=0$, but the zero-labeled leaf $i'$ has only the
zero-labeled neighbor $i$. Hence changing $f(w)$ from two to one preserves all
PID constraints and decreases the weight, a contradiction.

Now suppose that some $v\in K$ has label zero. The PID condition at $v$ gives
\begin{equation}
 f(x)=0\quad\text{for every }x\in K\setminus\{w\}
 \text{ and every }x\in N_G(v)\cap I.                 \label{eq:k11-zero}
\end{equation}
Every pair $\{i,i'\}$ has total label at least two. Indeed, this is immediate
when $f(i)=2$; if $f(i)=1$, then $f(i')\geq1$; and if $f(i)=0$, then $i$ is
not adjacent to $w$, so \eqref{eq:k11-zero} forces all its clique neighbors
to have label zero and therefore $f(i')=2$. Consequently,
\begin{equation}
 \omega(f)\geq 2+2|I|.                                \label{eq:k11-pair-bound}
\end{equation}

If a zero-labeled clique vertex $v$ has a neighbor in $I$, define $g$ by
$g(w)=g(v)=1$ and, for every $i\in N_G(v)\cap I$, set $g(i')=1$; retain all
other labels. Each such center $i$ is zero, is not adjacent to $w$, and had
$f(i')=2$. Thus its neighbor sum remains two. At every remaining zero-labeled
clique vertex, the contribution two from $w$ is replaced by one from $w$ and
one from $v$. Hence $g$ is a PID-function of smaller weight, a contradiction.

It remains that no zero-labeled clique vertex has a neighbor in $I$. Since
every center has a clique neighbor, every center is adjacent to $w$ and to no
other clique vertex. Choose a zero-labeled $v\in K$ and define $h$ by
\[
 h(w)=h(v)=1,\quad h(x)=0\ (x\in K\setminus\{w,v\}),
 \quad h(i)=0,\quad h(i')=1\ (i\in I).
\]
Then every zero-labeled clique vertex receives one from each of $w$ and $v$,
and every center receives one from $w$ and one from its leaf. Thus $h$ is a
PID-function with $\omega(h)=|I|+2<2|I|+2\leq\omega(f)$, again a
contradiction.
\qed
\end{proof}

\subsection{Proof of Theorem~\ref{thm:k11-hard}}\label{app:k11-hard}
\noindent\textbf{Theorem~\ref{thm:k11-hard}.}
\textit{PID is NP-complete on strict $K_{1,1}$-split graphs.}

\begin{proof}
Membership in NP is immediate. Let $(X,\mathcal C)$ be an RXC3 instance, where
$X=\{x_1,\ldots,x_{3q}\}$ and
$\mathcal C=\{S_1,\ldots,S_{3q}\}$. Construct a graph
$G=(K\cup I\cup J,E)$ as follows.
\begin{itemize}
    \item Let $K=X\cup\{z\}$ and make $K$ a clique.
    \item For each $j\in\{1,\ldots,3q\}$, introduce two center--leaf pairs
    $s_j u_j$ and $w_j y_j$, where
    \[
       I=\{s_j,w_j:1\leq j\leq3q\},\qquad
       J=\{u_j,y_j:1\leq j\leq3q\}.
    \]
    \item Join $s_j$ to $x_i$ exactly when $x_i\in S_j$.
    \item Join every $w_j$ to $z$. For $j\in\{1,2,3\}$, also join $w_j$
    to every vertex of $X$.
\end{itemize}
The graph is a strict $K_{1,1}$-split graph and has $15q+1$ vertices. Set
$B=9q+1$. Figure~\ref{fig:k11-reduction} shows the two types of gadgets, and
Figure~\ref{fig:k11-full-reduction} shows the complete reduced graph for a
small RXC3 instance.

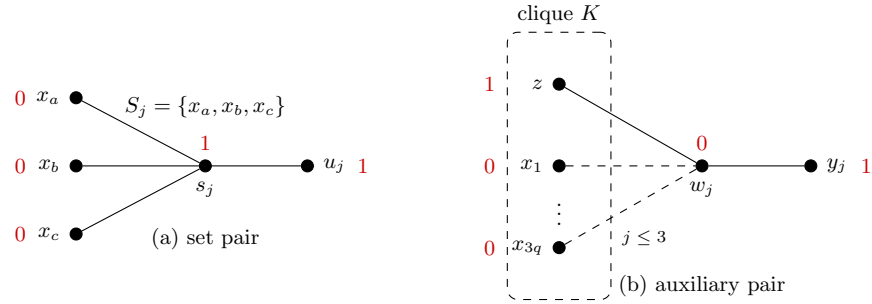
\begin{figure}[t]
\centering
\begin{tikzpicture}[scale=0.9,transform shape,node distance=9mm and 14mm]
  \begin{scope}[xshift=-3.7cm]
    \node[vertexdot] (xa) at (-1.7,1) {};
    \node[vertexdot] (xb) at (-1.7,0) {};
    \node[vertexdot] (xc) at (-1.7,-1) {};
    \node[vertexdot] (sj) at (0.2,0) {};
    \node[vertexdot] (uj) at (1.7,0) {};
    \node[vertexlabel,left=1mm of xa] {$x_a$};
    \node[vertexlabel,left=1mm of xb] {$x_b$};
    \node[vertexlabel,left=1mm of xc] {$x_c$};
    \node[vertexlabel,below=1mm of sj] {$s_j$};
    \node[vertexlabel,right=1mm of uj] {$u_j$};
    \node[pidvalue,left=6mm of xa] {$0$};
    \node[pidvalue,left=6mm of xb] {$0$};
    \node[pidvalue,left=6mm of xc] {$0$};
    \node[pidvalue,above=1mm of sj] {$1$};
    \node[pidvalue,right=6mm of uj] {$1$};
    \draw (xa)--(sj) (xb)--(sj) (xc)--(sj) (sj)--(uj);
    \node[above=5mm of sj] {$S_j=\{x_a,x_b,x_c\}$};
    \node[below=7mm of sj] {(a) set pair};
  \end{scope}

  \begin{scope}[xshift=3.3cm]
    \node[vertexdot] (z) at (-1.6,1.2) {};
    \node[vertexdot] (xone) at (-1.6,0) {};
    \node[vertexdot] (xlast) at (-1.6,-1.2) {};
    \node[vertexlabel,left=1mm of z] {$z$};
    \node[vertexlabel,left=1mm of xone] {$x_1$};
    \node[vertexlabel,left=1mm of xlast] {$x_{3q}$};
    \node[pidvalue,left=8mm of z] {$1$};
    \node[pidvalue,left=8mm of xone] {$0$};
    \node[pidvalue,left=8mm of xlast] {$0$};
    \node at (-1.6,-0.6) {$\vdots$};
    \node[draw,dashed,rounded corners,fit=(z)(xone)(xlast),inner sep=6mm,
      label=above:{clique $K$}] {};
    \node[vertexdot] (wj) at (0.5,0) {};
    \node[vertexdot] (yj) at (2.1,0) {};
    \node[vertexlabel,below=1mm of wj] {$w_j$};
    \node[vertexlabel,right=1mm of yj] {$y_j$};
    \node[pidvalue,above=1mm of wj] {$0$};
    \node[pidvalue,right=6mm of yj] {$1$};
    \draw (z)--(wj) (wj)--(yj);
    \draw[dashed] (xone)--(wj) (xlast)--(wj);
    \node[font=\scriptsize] at (-0.35,-1.05) {$j\leq3$};
    \node at (0.5,-1.75) {(b) auxiliary pair};
  \end{scope}
\end{tikzpicture}
\caption{Hardness construction for strict $K_{1,1}$-split graphs, with PID
values shown in red. The set pair is displayed in its selected mode. The
center $s_j$ is adjacent to the three element vertices representing $S_j$.
Every $w_j$ is adjacent to $z$; for $j\leq3$, it is also adjacent to every
vertex of $X$. Each center has one private leaf.}
\label{fig:k11-reduction}
\end{figure}

\begin{figure}[t]
\centering
\begin{tikzpicture}[scale=0.82,transform shape,
  every node/.style={font=\scriptsize}]
  \node[vertexdot] (Kz) at (-4.8,3.6) {};
  \node[vertexdot] (Kx1) at (-4.8,2.7) {};
  \node[vertexdot] (Kx2) at (-4.8,1.8) {};
  \node[vertexdot] (Kx3) at (-4.8,0.9) {};
  \node[vertexdot] (Kx4) at (-4.8,0) {};
  \node[vertexdot] (Kx5) at (-4.8,-0.9) {};
  \node[vertexdot] (Kx6) at (-4.8,-1.8) {};
  \node[vertexlabel,left=1mm of Kz] {$z$};
  \node[vertexlabel,left=1mm of Kx1] {$x_1$};
  \node[vertexlabel,left=1mm of Kx2] {$x_2$};
  \node[vertexlabel,left=1mm of Kx3] {$x_3$};
  \node[vertexlabel,left=1mm of Kx4] {$x_4$};
  \node[vertexlabel,left=1mm of Kx5] {$x_5$};
  \node[vertexlabel,left=1mm of Kx6] {$x_6$};
  \node[pidvalue,left=9mm of Kz] {$1$};
  \foreach \x in {Kx1,Kx2,Kx3,Kx4,Kx5,Kx6}
    \node[pidvalue,left=9mm of \x] {$0$};
  \node[draw,dashed,rounded corners,
    fit=(Kz)(Kx1)(Kx2)(Kx3)(Kx4)(Kx5)(Kx6),inner sep=5mm,
    label=above:{ }] {};

  \foreach \n/\lab/\val/\yy in {
    w1/w_1/0/3.6,w2/w_2/0/3.05,w3/w_3/0/2.5,
    w4/w_4/0/1.95,w5/w_5/0/1.4,w6/w_6/0/0.85,
    s1/s_1/0/0.05,s2/s_2/1/-0.5,s3/s_3/0/-1.05,
    s4/s_4/1/-1.6,s5/s_5/0/-2.15,s6/s_6/0/-2.7}
  {
    \node[vertexdot] (I\n) at (0,\yy) {};
    \node[vertexlabel,above left=0mm and 0.8mm of I\n] {$\lab$};
    \node[pidvalue,above right=0mm and 0.8mm of I\n] {$\val$};
  }

  \foreach \n/\lab/\val/\yy in {
    y1/y_1/1/3.6,y2/y_2/1/3.05,y3/y_3/1/2.5,
    y4/y_4/1/1.95,y5/y_5/1/1.4,y6/y_6/1/0.85,
    u1/u_1/2/0.05,u2/u_2/1/-0.5,u3/u_3/2/-1.05,
    u4/u_4/1/-1.6,u5/u_5/2/-2.15,u6/u_6/2/-2.7}
  {
    \node[vertexdot] (J\n) at (4.2,\yy) {};
    \node[vertexlabel,right=1mm of J\n] {$\lab$};
    \node[pidvalue,right=9mm of J\n] {$\val$};
  }

  \foreach \a/\b in {w1/y1,w2/y2,w3/y3,w4/y4,w5/y5,w6/y6,
    s1/u1,s2/u2,s3/u3,s4/u4,s5/u5,s6/u6}
    \draw (I\a)--(J\b);

  \foreach \w in {Iw1,Iw2,Iw3}
    \foreach \x in {Kz,Kx1,Kx2,Kx3,Kx4,Kx5,Kx6}
      \draw[blue!55] (\x)--(\w);
  \foreach \w in {Iw4,Iw5,Iw6}
    \draw[blue!55] (Kz)--(\w);

  \draw (Kx1)--(Is1) (Kx2)--(Is1) (Kx3)--(Is1);
  \draw (Kx1)--(Is2) (Kx2)--(Is2) (Kx4)--(Is2);
  \draw (Kx1)--(Is3) (Kx5)--(Is3) (Kx6)--(Is3);
  \draw (Kx3)--(Is4) (Kx5)--(Is4) (Kx6)--(Is4);
  \draw (Kx2)--(Is5) (Kx4)--(Is5) (Kx5)--(Is5);
  \draw (Kx3)--(Is6) (Kx4)--(Is6) (Kx6)--(Is6);

  \node[font=\bfseries] at (-4.8,4.45) {$K$};
  \node[font=\bfseries] at (0,4.45) {$I$};
  \node[font=\bfseries] at (4.2,4.45) {$J$};
  \node[align=center] at (0,-3.5)
    {Blue edges are just used for better visualization.};
\end{tikzpicture}
\caption{The full strict $K_{1,1}$-split graph produced by the reduction for
an illustrative RXC3 instance with $X=\{x_1,\ldots,x_6\}$. The red PID values
correspond to the exact cover $\{S_2,S_4\}$. The clique edges inside $K$ are
suppressed to keep the drawing readable. The centers $s_1,\ldots,s_6$ encode
the six sets, while $w_1,w_2,w_3$ are adjacent to all of $K$ and
$w_4,w_5,w_6$ are adjacent only to $z$. Every center has its private leaf in
$J$.}
\label{fig:k11-full-reduction}
\end{figure}
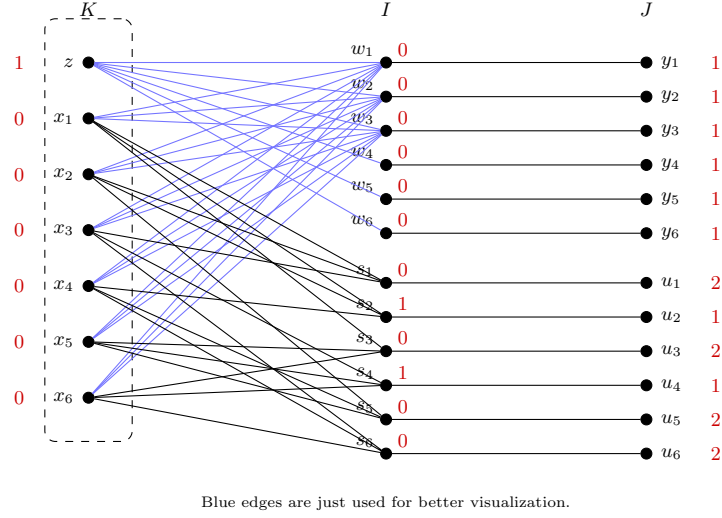

Suppose first that $\mathcal C^*$ is an exact cover. Assign label one to $z$,
to every $s_j,u_j$ with $S_j\in\mathcal C^*$, and to every $y_j$. Assign
label two to $u_j$ when $S_j\notin\mathcal C^*$, and assign zero to all
remaining vertices. Every element vertex receives one from $z$ and one from
the unique selected set containing it; every zero-labeled set center receives
two from its leaf; and every auxiliary center receives one from $z$ and one
from its leaf. Hence this is a PID-function of weight $9q+1$.

Conversely, suppose that $G$ has a PID-function of weight at most $B$, and
choose such a function $f$ of minimum weight. By
Lemma~\ref{lem:no-two-clique}, every clique label is zero or one.

We first show that $f_K=1$. If $f_K>2$, then no clique vertex is zero, since a
zero-labeled clique vertex is adjacent to all other clique vertices and would
receive clique contribution greater than two. Thus every clique vertex has
label one. Each pair $\{s_j,u_j\}$ then has weight at least two, and the same
is true for the three pairs $\{w_j,y_j\}$ with $j\leq3$. Therefore
$\omega(f)\geq(3q+1)+2(3q)+6=9q+7>B$.
If $f_K=2$, exactly two clique vertices have label one. Any zero-labeled
clique vertex already receives two from the clique, so $w_1$ must be zero.
The center $w_1$ also receives two from the clique and therefore forces
$y_1=0$, but then the zero-labeled leaf $y_1$ has neighbor sum zero, a
contradiction. If $f_K=0$, every center--leaf pair has total label at least
two, giving $\omega(f)\geq12q>B$. Hence $f_K=1$.

Let $h$ be the unique clique vertex with label one. We claim that $h=z$. If
$h=x_t$ for some element vertex, the PID condition at $z$ gives
$\sum_{j}f(w_j)=1$. The auxiliary pairs then contribute at least $6q-3$: the
first three pairs contribute at least one each, and all pairs with $j\geq4$
contribute at least two. Exactly three sets contain $x_t$. Their set pairs
contribute at least one each, while every other set pair contributes at least
two, for a total of at least $6q-3$. Consequently,
$\omega(f)\geq1+(6q-3)+(6q-3)=12q-5>9q+1$,
where the last inequality uses $q\geq3$. Thus $f(z)=1$ and every element
vertex has label zero.

With these clique labels, every auxiliary pair has weight at least one and
every set pair has weight at least two. Hence
$\omega(f)\geq1+3q+2(3q)=9q+1$.
Equality must hold throughout. Therefore $f(w_j)=0$ and $f(y_j)=1$ for every
$j$, and every set pair has total label two. For each element $x_i$, the PID
condition now reads
$1+\sum_{j:x_i\in S_j}f(s_j)=2$.
Thus exactly one set center adjacent to $x_i$ has label one. It follows that
$\mathcal C^*=\{S_j:f(s_j)=1\}$
is an exact cover of $X$. This proves the equivalence and the theorem.
\qed
\end{proof}

\subsection{Proof of Theorem~\ref{thm:strict-stars}}\label{app:strict-stars}
\noindent\textbf{Theorem~\ref{thm:strict-stars}.}
\textit{On strict $K_{1,r}$-split graphs, PID is NP-complete for $r=1,2$.
For every fixed $r\geq3$, a minimum PID-function can be computed in
$O(n+m)$ time, where $n=|V(G)|$ and $m=|E(G)|$, without a balance
assumption.}

\begin{proof}
For $r=1$, use Theorem~\ref{thm:k11-hard}. For $r=2$, NP-completeness
is known from Pradhan, Banerjee, and Liu~\cite{PID_CSPB}.

Now let $r\geq3$. Recover a partition using
Lemma~\ref{lem:strict-recognition}, and put $n=|V(G)|$, $m=|E(G)|$,
$k=|K|$, and $t=|I|$. Every center must be positive: if it were zero,
all its leaves would be positive and contribute at least $r\geq3$ to its
neighborhood sum. A center of label two has minimum star cost two, attained
with all leaves zero. A center of label one has minimum star cost $r+1$,
attained with all leaves one. This normalization affects no zero vertex
other than leaves with a center of label two, whose constraints are satisfied.

For $v\in K$, compute $d(v)=|N(v)\cap I|$, and write
$s=\sum_{v\in K}f(v)$. If a clique vertex is zero, then $s\leq2$.
We compute four candidate weights, using $+\infty$ for an infeasible case.

\smallskip
\noindent\emph{All clique vertices positive.}
The minimum weight in this case is $W_+=k+2t$: label every clique vertex one,
every center two, and every leaf zero. Each positive clique vertex costs at
least one and each star costs at least two, so this is optimal within the case.

\smallskip
\noindent\emph{Clique weight $s=0$.}
Every clique vertex is zero and its center neighbors must sum to two.
Thus $d(v)\in\{1,2\}$ for every $v\in K$. A vertex with $d(v)=1$ forces
its unique center neighbor to two, while one with $d(v)=2$ forces both center
neighbors to one. Let $A_2$ and $A_1$ be the respective sets of forced centers.
If some $d(v)\notin\{1,2\}$ or $A_1\cap A_2\neq\emptyset$, the case is
infeasible. Otherwise label $A_1$ one and every other center two, giving
$W_0=2t+(r-1)|A_1|$.
These labels meet all constraints, and the forced labels show optimality.

\smallskip
\noindent\emph{Clique weight $s=1$.}
Let $w$ be the unique clique vertex labeled one. Every other clique vertex
$v$ must have $d(v)=1$, forcing its sole center neighbor to one. Hence $w$
is eligible exactly when it contains all exceptions to $d(v)=1$: if there
are no exceptions, every $w\in K$ is eligible; if there is one, only that
vertex is eligible; and if there are two or more, this case is infeasible.
Compute, for every $w\in K$,
$p(w)=|\{u\in I:N(u)\cap K=\{w\}\}|$.
Each center has a clique neighbor. The centers counted by $p(w)$ are
therefore precisely the unforced centers; give them label two, and give all
other centers label one. For the eligible set $Q\subseteq K$, the optimum is
$W_1=1+2t+(r-1)\bigl(t-\max_{w\in Q}p(w)\bigr)$,
with $W_1=+\infty$ if $Q=\emptyset$. Choose a maximizing $w$ and normalize
the leaves. Every zero clique vertex then receives one from $w$ and one
from its unique center neighbor.

\smallskip
\noindent\emph{Clique weight $s=2$.}
A zero clique vertex can have no center neighbor, since all centers are
positive. Thus $H=\{v\in K:d(v)>0\}$ must lie in the positive clique
support, which has size at most two. If $|H|>2$, the case is infeasible.
Otherwise set $W_2=2+2t$: for $|H|=2$ label its vertices one; for $|H|=1$
label its vertex two; and for $H=\emptyset$ label any clique vertex two.
Label all other clique vertices zero, all centers two, and all leaves zero.
The clique is nonempty, so this construction is always defined. Every zero
clique vertex sees exactly two from the clique and nothing from the centers.
The lower bounds on clique and star weights prove optimality in this case.

The cases are exhaustive (and may overlap when all clique vertices are
positive), so the optimum is $\min\{W_+,W_0,W_1,W_2\}$. The degrees $d(v)$,
forced-center sets, exception count, and values $p(w)$ are obtained by a
constant number of adjacency-list scans. The maximum over eligible vertices
is a single scan, not an enumeration of complete labelings. Store the winning
case and, when needed, its vertex $w$; reconstruct its center and leaf labels
in one further scan. Including partition recovery, the running time is
$O(n+m)$, with no balance assumption.
\qed
\end{proof}

\subsection{Proof of Theorem~\ref{thm:split-plus-one-inapprox}}\label{app:split-plus-one-inapprox}
\noindent\textbf{Theorem~\ref{thm:split-plus-one-inapprox}.}
\textit{For every fixed $0<\varepsilon<1$, minimum perfect Italian domination has
no polynomial-time $(2-\varepsilon)$-approximation algorithm on graphs $G$
with $\operatorname{sd}(G)=1$, unless $\mathrm{P}=\mathrm{NP}$.}

\begin{proof}
Fix $0<\varepsilon<1$. We reduce from Restricted Exact Cover by 3-Sets
(RXC3)~\cite{gonzalez1985clustering}. Let $U$ have $3q$ elements, and let
$\mathcal F$ consist of $3q$ three-element subsets of $U$, with every
element occurring in exactly three members of $\mathcal F$. We may assume
$q\geq2$ and $\varepsilon(q+2)>2$: if necessary, replace the RXC3 instance
by a fixed number of disjoint copies with distinct element names. This
preserves the occurrence restriction and the existence of an exact cover;
the number of copies depends only on the fixed constant $\varepsilon$.

Use the construction of Theorem~\ref{thm:split-plus-one-hardness}, changing
only the number of test vertices to $L=2q+3$. Explicitly, let
\[
 K_U=\{k_u:u\in U\},\qquad T=\{t_1,\ldots,t_L\},\qquad
 X=\{x_F:F\in\mathcal F\}.
\]
Make $C=\{w\}\cup K_U\cup T$ a clique and $X$ an independent set.
Add a vertex $s$ adjacent exactly to $T\cup X$. Join $w$ to every vertex
of $X$, and join $x_F$ to $k_u$ exactly when $u\in F$. No other edges
are added. The resulting graph is connected and has $8q+5$ vertices.
Deleting $s$ leaves the split partition $(C,X)$. For any $F\in\mathcal F$,
$u\in F$, and $t\in T$, the vertices $s,t,k_u,x_F$ induce a four-cycle
in that order. Thus $G$ is not split and $\operatorname{sd}(G)=1$.

We prove the gap
\begin{equation}
 \begin{aligned}
  (U,\mathcal F)\text{ has an exact cover}
    &\ \Longrightarrow\ \gamma_I^P(G)\leq q+2,\\
  (U,\mathcal F)\text{ has no exact cover}
    &\ \Longrightarrow\ \gamma_I^P(G)\geq2q+2.
 \end{aligned}
 \label{eq:split-plus-one-gap}
\end{equation}
For the first implication, label $s$ and $w$ one, label the $q$ set vertices
of an exact cover one, and label all other vertices zero. A zero test
vertex receives one from each of $s$ and $w$; a zero element vertex receives
one from $w$ and one from its selected set vertex; and an unselected set
vertex receives one from each of $s$ and $w$. This is a PID-function of
weight $q+2$, regardless of the enlarged number of test vertices.

For the second implication, suppose that $f$ is a PID-function with
$\omega(f)<2q+2$. Put $c=f(C)$ and $x=f(X)$. Since $|T|=2q+3$, at least
one test vertex is zero. Its PID condition gives
$f(s)+c=2$.
We show that such a function forces an exact cover.

\smallskip
\noindent\emph{Case 1: $f(s)=0$.}
Then $c=2$ and the PID condition at $s$ gives $f(T)+x=2$. At most two
element vertices are positive, so every triple $F$ contains a zero-labeled
element vertex. This vertex already receives two from $C$ and therefore
forces $f(x_F)=0$. Hence $x=0$, so $f(T)=2$ and $w$ and all element
vertices are zero. Every set vertex is then zero with neighborhood sum
zero, a contradiction.

\smallskip
\noindent\emph{Case 2: $f(s)=2$.}
Then $c=0$. Each element vertex requires
$\sum_{F\in\mathcal F:\,u\in F}f(x_F)=2$. Summing over $U$ gives $3x=6q$, because every
set has three elements. Thus
$\omega(f)\geq f(s)+x=2+2q$, contrary to the assumed strict inequality.

\smallskip
\noindent\emph{Case 3: $f(s)=1$.}
Then $c=1$, so exactly one vertex $h\in C$ has label one. If $h=k_y$
is an element vertex, then $f(w)=0$ and its PID condition gives $1+x=2$.
Exactly one set vertex consequently has label one. Each of the remaining
$3q-1\geq5$ zero element vertices requires one unit from its incident set
vertices, but a single set vertex has only three element neighbors. This
is impossible. Therefore $h\in\{w\}\cup T$, and all element vertices
are zero. Their PID conditions give
$\sum_{F\in\mathcal F:\,u\in F}f(x_F)=1$ for every $u\in U$.
Every positive set vertex has label one, and the corresponding triples
cover each element exactly once. They form an exact cover. This proves
the second implication of~\eqref{eq:split-plus-one-gap}.

Suppose that a polynomial-time $(2-\varepsilon)$-approximation algorithm
exists on the stated graph class. On a yes-instance, its feasible output
has weight at most
\[
 (2-\varepsilon)(q+2)
 =2q+4-\varepsilon(q+2)<2q+2.
\]
On a no-instance, every feasible output has weight at least $2q+2$.
Comparing the returned weight with $2q+2$ therefore decides RXC3 in
polynomial time. The construction has polynomial size, so this contradicts
$\mathrm{P}\ne\mathrm{NP}$ and proves the theorem.
\qed
\end{proof}

\subsection{Proof of Theorem~\ref{thm:no-ptas}}\label{app:no-ptas}
\noindent\textbf{Theorem~\ref{thm:no-ptas}.}
\textit{Minimum PID on strict $K_{1,1}$-split graphs does not admit a PTAS unless $\mathrm{P}=\mathrm{NP}$. This holds even when $|K|\leq|I|$.}

\begin{proof}
Start with an instance of the gap problem described above.
If necessary, first replace it by a fixed number of disjoint copies of the
\emph{matching instance}, using distinct element names. This multiplies
$q$, $m$, and $\nu(\mathcal T)$ by the same factor, preserves the occurrence
bound and both gap cases, and lets us assume $q>2B+1$.

Construct a strict $K_{1,1}$-split graph as follows. Let $K=U\cup\{z\}$ and
make $K$ a clique. For every triple $T\in\mathcal T$, add a center $s_T$ with
a private leaf, join $s_T$ to $z$, and join it to the three element vertices
of $T$. For every $x\in U$, add a private center $p_x$ with one leaf and join
$p_x$ to $z$ and $x$. Thus there are
 $N=m+3q$
center--leaf pairs.

The graph is connected because every center is adjacent to $z$. Also,
$|K|=3q+1\leq m+3q=|I|$ by~\eqref{eq:bd3dm-size}. We claim that
\begin{equation}
 \gamma_I^P(G)=1+m+6q-2\nu(\mathcal T).          \label{eq:gap-optimum}
\end{equation}
Let $\mathcal M$ be a matching of size $t$. Assign one to $z$, to every set
center representing a triple of $\mathcal M$, to the private center of every
element not covered by $\mathcal M$, and to every leaf. Assign zero elsewhere.
Every element vertex receives one from $z$ and one from exactly one selected
center. This PID-function has weight
 $1+N+t+(3q-3t)=1+m+6q-2t.$

For the reverse inequality, choose a minimum PID-function $f$. By
Lemma~\ref{lem:no-two-clique}, clique labels belong to $\{0,1\}$. If
$f_K=0$, all $N$ center--leaf pairs have weight at least two. If the unique
clique vertex of label one is an element $x$, at most $B+1$ centers are
adjacent to it, and hence
 $\omega(f)\geq 2N-B.$

If $f_K=2$ and $z$ is positive, the other positive clique vertex belongs to a
triple whose center is forced to zero by one of its two zero-labeled element
vertices; that center then receives two from the clique and yields an
impossible zero center--zero leaf pair. If $f_K=2$ and $z$ is zero, all centers
are zero, and at most $2B+2$ pairs can have weight one. Therefore,
 $\omega(f)\geq2N-2B.$

Finally, if $f_K>2$, every clique vertex is positive and every pair has weight
at least two, so $\omega(f)\geq3q+1+2N$.

By~\eqref{eq:bd3dm-size} and $q>2B+1$, the smallest of the first three
weight bounds satisfies
$(2N-2B)-(1+m+6q)=m-2B-1>0$,
and the all-positive-clique bound exceeds $1+m+6q$ by $m+3q>0$.
However, the construction from the empty matching gives a PID-function of
weight $1+m+6q$. Thus a minimum solution has $f(z)=1$ and every element
vertex has label zero. Each center has an element neighbor whose PID
constraint allows only one further unit beyond the contribution from $z$;
hence every center has label zero or one. A zero center must have a leaf of
label one, and a center of label one forces its leaf to be positive, with
label one optimal. Therefore every leaf has label one in a minimum solution.
The PID equation at each element requires exactly one incident center of
label one. The selected triple centers represent a matching, and precisely
the uncovered elements use their private centers. If $t$ triple centers are
selected, the weight is $1+m+6q-2t$. Taking a maximum matching proves
\eqref{eq:gap-optimum}.

Apply Petrank's gap with the same constants $B$ and $\delta$.
By~\eqref{eq:gap-optimum}, a perfect matching gives PID optimum
$A=1+m+4q$, whereas the no-case gives optimum at least $A+2\delta q$.
Since $m\leq Bq$ and $q\geq1$, we have $A\leq(B+5)q$. Set
$\varepsilon=\frac{\delta}{B+5}>0$.
A polynomial-time $(1+\varepsilon)$-approximation would return a feasible
PID-function of weight at most $A+\delta q$ in the yes-case. In the no-case,
every feasible PID-function has weight at least $A+2\delta q$. Comparing
the returned weight with $A+\tfrac32\delta q$ would distinguish the two
cases in polynomial time, a contradiction unless $\mathrm{P}=\mathrm{NP}$.
This proves the claimed inapproximability on the balanced subclass.
\qed
\end{proof}

\subsection{Proof of Proposition~\ref{prop:balanced-approx}}\label{app:balanced-approx}
\noindent\textbf{Proposition~\ref{prop:balanced-approx}.}
\textit{Let $G$ be a strict $K_{1,r}$-split graph with $r\in\{1,2\}$ admitting a
partition $(K,I,J)$ with $|K|\leq c|I|$ for a fixed constant $c$.
There is a linear-time $(c+2)$-approximation for $r=1$ and a linear-time
$(1+c/2)$-approximation for $r=2$.}

\begin{proof}
Use Lemma~\ref{lem:strict-recognition} to recover a partition in linear time.
Its part sizes equal those of any defining partition, so it satisfies
$|K|\leq c|I|$. Assign label one to every vertex of $K$, label two to every
center, and label zero to every leaf. This takes linear time and gives a
PID-function of weight $|K|+2|I|$. For $r=1$,
every center--leaf pair has total weight at least one in every PID-function, so
$\gamma_I^P(G)\geq|I|$. This gives ratio at most $c+2$. For $r=2$, every
star component has weight at least two: a center labeled two already has
weight two; a center labeled one forces all leaves to be positive; and a
zero-labeled center forces every leaf to be positive. Thus
$\gamma_I^P(G)\geq2|I|$, giving ratio at most $1+c/2$.
\qed
\end{proof}


\begin{thebibliography}{99}
\bibitem{roman}
ReVelle, C.; Rosing, K.
Defendens Imperium Romanum: A Classical Problem in Military Strategy.
{\em Am. Math. Mon.} {\bf 2000}, {\em 107}, 585--594.
\url{https://doi.org/10.1080/00029890.2000.12005243}.

\bibitem{pagourtzis2002server}
Pagourtzis, A.; Penna, P.; Schlude, K.; Steinhel, K.; Taylor, D.S.; Widmayer, P.
Server placements, Roman domination and other dominating set variants.
In \emph{Proceedings of the IFIP TCS Conference}, 2002; Vol. 271, pp. 280--291.

\bibitem{ghaffari_wsn}
Ghaffari, F.; Bahrak, B.; Shariatpanahi, S.P.
A novel approach to partial coverage in wireless sensor networks via the Roman dominating set.
{\em IET Netw.} {\bf 2022}, {\em 11}, 58--69.
\url{https://doi.org/10.1049/ntw2.12034}.

\bibitem{roman_domination_graphs}
Cockayne, E.J.; Dreyer, P.A., Jr.; Hedetniemi, S.M.; Hedetniemi, S.T.
Roman domination in graphs.
{\em Discret. Math.} {\bf 2004}, {\em 278}, 11--22.
\url{https://doi.org/10.1016/j.disc.2003.06.004}.

\bibitem{HENNING2018235}
Henning, M.A.; Klostermeyer, W.F.; MacGillivray, G.
Perfect Roman domination in trees.
{\em Discret. Appl. Math.} {\bf 2018}, {\em 236}, 235--245.
\url{https://doi.org/10.1016/j.dam.2017.10.027}.

\bibitem{italian_domination}
Chellali, M.; Haynes, T.W.; Hedetniemi, S.T.; McRae, A.A.
Roman $\{2\}$-domination.
{\em Discret. Appl. Math.} {\bf 2016}, {\em 204}, 22--28.
\url{https://doi.org/10.1016/j.dam.2015.11.013}.

\bibitem{varghese2021italian}
Varghese, J.; Anu, V.; Lakshmanan, A.
Italian domination and perfect Italian domination on Sierpiński graphs.
{\em J. Discrete Math. Sci. Cryptogr.} {\bf 2021}, {\em 24}, 1885--1894.
\url{https://doi.org/10.1080/09720529.2021.1933705}.

\bibitem{almulhim2024survey}
Almulhim, A.; Al Subaiei, B.; Mondal, S.R.
Survey on Roman $\{2\}$-domination.
{\em Mathematics} {\bf 2024}, {\em 12}, 2771.
\url{https://doi.org/10.3390/math12172771}.

\bibitem{perfect_italian_domination}
Haynes, T.W.; Henning, M.A.
Perfect Italian domination in trees.
{\em Discret. Appl. Math.} {\bf 2019}, {\em 260}, 164--177.
\url{https://doi.org/10.1016/j.dam.2019.01.038}.

\bibitem{EDK1K2}
Brandstädt, A.; Mosca, R.
On efficient domination for some classes of $H$-free chordal graphs.
{\em Discret. Appl. Math.} {\bf 2020}, {\em 281}, 81--95.
\url{https://doi.org/10.1016/j.dam.2019.04.029}.

\bibitem{PRD_CBIDCO}
Mann, K.; Fernau, H.
Perfect Roman Domination: Aspects of enumeration and parameterization.
\emph{Algorithms} \textbf{2024}, \emph{17}, 576.
\url{https://doi.org/10.3390/a17120576}.

\bibitem{PRD_BCPPISPBC}
Banerjee, S.; Keil, J.M.; Pradhan, D.
Perfect Roman domination in graphs.
\emph{Theoretical Computer Science} \textbf{2019}, \emph{796}, 1--21.
\url{https://doi.org/10.1016/j.tcs.2019.08.017}.

\bibitem{PID_CSPB}
Pradhan, D.; Banerjee, S.; Liu, J.-B.
Perfect Italian domination in graphs: Complexity and algorithms.
\emph{Discrete Applied Mathematics} \textbf{2022}, \emph{319}, 271--295.
\url{https://doi.org/10.1016/j.dam.2021.08.020}.

\bibitem{PRD_S}
Fernau, H.; Mann, K.
Perfect Roman domination and unique response Roman domination.
\emph{arXiv} \textbf{2023}, arXiv:2309.07018.

\bibitem{ID_SB}
Chen, H.; Lu, C.
Roman $\{2\}$-domination problem in graphs.
\emph{Discussiones Mathematicae Graph Theory} \textbf{2022}, \emph{42},
641--660.
\url{https://doi.org/10.7151/dmgt.2332}.

\bibitem{PID_DCCP45}
Paul, K.; Pandey, A.
Perfect Italian domination on some generalizations of cographs.
\emph{Computational and Applied Mathematics} \textbf{2024}, \emph{43}(6),
390.
\url{https://doi.org/10.1007/s40314-024-02901-5}.

\bibitem{lauri2020perfect}
Lauri, J.; Mitillos, C.
Perfect Italian domination on planar and regular graphs.
{\em Discret. Appl. Math.} {\bf 2020}, {\em 285}, 676--687.
\url{https://doi.org/10.1016/j.dam.2020.05.024}.

\bibitem{PID_Cograph}
Banerjee, S.; Henning, M.A.; Pradhan, D.
Perfect Italian domination in cographs.
\emph{Applied Mathematics and Computation} \textbf{2021}, \emph{391}, 125703.
\url{https://doi.org/10.1016/j.amc.2020.125703}.

\bibitem{ford1956maxflow}
Ford, L.R., Jr.; Fulkerson, D.R.
Maximal flow through a network.
\emph{Can. J. Math.} \textbf{1956}, \emph{8}, 399--404.
\url{https://doi.org/10.4153/CJM-1956-045-5}.

\bibitem{edmondskarp1972}
Edmonds, J.; Karp, R.M.
Theoretical improvements in algorithmic efficiency for network flow problems.
\emph{J. ACM} \textbf{1972}, \emph{19}, 248--264.
\url{https://doi.org/10.1145/321694.321699}.

\bibitem{gonzalez1985clustering}
Gonzalez, T.F.
Clustering to minimize the maximum intercluster distance.
{\em Theor. Comput. Sci.} {\bf 1985}, {\em 38}, 293--306.
\url{https://doi.org/10.1016/0304-3975(85)90224-5}.

\bibitem{kann1991}
Kann, V.
Maximum bounded 3-dimensional matching is MAX SNP-complete.
\emph{Inf. Process. Lett.} \textbf{1991}, \emph{37}, 27--35.
\url{https://doi.org/10.1016/0020-0190(91)90246-E}.
\bibitem{petrank1994gap}
Petrank, E.
The hardness of approximation: Gap location.
\emph{Computational Complexity} \textbf{1994}, \emph{4}, 133--157.
\url{https://doi.org/10.1007/BF01202286}.
\end{thebibliography}
\end{document}